\documentclass[a4paper]{article}

\usepackage{amsmath,amssymb,mathrsfs,bbm,graphicx,caption}
\usepackage[amsmath,thmmarks,thref]{ntheorem}

\allowdisplaybreaks[3]

\theoremstyle{plain}
\newtheorem{theorem}{Theorem}

\newtheorem{proposition}[theorem]{Proposition}
\newtheorem{lemma}[theorem]{Lemma}
\newtheorem{corollary}[theorem]{Corollary}
\newtheorem{conjecture}[theorem]{Conjecture}

\theorembodyfont{\rmfamily}
\newtheorem{remark}[theorem]{Remark}
\newtheorem{example}[theorem]{Example}

\theoremstyle{nonumberplain}
\theoremsymbol{$\square$} \theoremheaderfont{\rmfamily\itshape} \theorembodyfont{\rmfamily}
\newtheorem{proof}{Proof.}

\def\de{\mathord{\rm d}}
\def\e{\mathord{\rm e}}

\DeclareMathOperator{\pr}{pr}

\DeclareMathOperator{\cov}{cov}
\DeclareMathOperator{\var}{var}
\DeclareMathOperator{\logit}{logit}

\begin{document}
\begin{center}
{\bfseries\LARGE Component Importance\\ Based on Dependence Measures\par}
\medskip\par
Mario Hellmich\footnote{mario.hellmich@bfe.bund.de}
\smallskip\par
{\small Bundesamt f\"ur kerntechnische Entsorgungssicherheit\\ (Federal Office for the Safety of Nuclear Waste Management) \\ Willy-Brandt-Stra\ss e~5, 38226 Salzgitter, Germany}\par
{Dated\kern1pt: \today} 
\end{center}

\begin{abstract}
\noindent We discuss the construction of component importance measures for binary coherent reliability systems from known stochastic dependence measures by measuring the dependence between system and component failures. We treat both the time-dependent case in which the system and its components are described by binary random variables at a fixed instant as well as the continuous time case where the system and component life times are random variables. As dependence measures we discuss covariance and mutual information, the latter being based on Shannon entropy. We prove some basic properties of the resulting importance measures and obtain results on importance ordering of components.\par
\medskip
\noindent\emph{Keywords\kern1pt:} Reliability theory, component importance measure, binary coherent system, stochastic dependence, entropy
\end{abstract}

\section{Introduction}
\label{sec::Introduction}
Component importance measures are used in reliability theory and engineering to rank the individual components of a system according to their influence on aspects of the system performance (mostly reliability), hence they serve to identify weaknesses and improvement potential. A large number of importance measures have been proposed, see the reviews~\cite{BEN95,BS01,KZ12,KZ12a} as well as the monograph~\cite{KZ12b}. In view of this multiplicity it is not surprising to find a variety approaches to construct  importance measures. A recurring idea is, however, to obtain them in some way from the connection between system failure and failure of an individual component, such as in the Birnbaum measure~\cite{B69} in the time-dependent case, or the Barlow--Proschan measure~\cite{BP75a} in the time-independent case. The present paper sets out from the question whether it is possible to obtain good importance measures by using known dependence measures and applying them to random variables associated with component and system failure. Thereby, from a more fundamental point of view, we hope to shed some light on the connection between component importance and dependence.\par
Generally, one can distinguish between time-dependent importance measures, which only depend on the system structure and the component reliabilities, which are considered either at a fixed instant or are assumed static (e.\,g.\ to calculate probabilities of failure on demand), and time-independent measures which take into account the component life distributions~\cite{BEN95}. In the former case, the system consisting of~$n$ components is described by binary random variables~$X_1,\ldots,X_n$ and~$X$, indicating the functioning or failure of the components and the system, respectively, either at a fixed point of time or independent of time. The latter case describes the system in continuous time by the nonnegative random variables~$T_1,\ldots,T_n$ and~$T$, which are the component and system failure instants, respectively. This is important since, in a continuous time setting, time-dependent measures leave it to the analyst at which time points to evaluate and compare them. In the present paper we shall address both situations. The idea is to use a dependence measure~$\mathscr D(U,V)\in\mathbb R$ defined for a pair of random variables~$U,V$ of a sufficiently rich class. We will investigate whether~$\mathscr D(X_i,X)$ in the binary case and~$\mathscr D(T_i,T)$ in the continuous time case, for $i=1,\ldots,n$, define sensible importance measures for particular choices of~$\mathscr D$.\par
It is not self-evident which properties one should require of a sensible dependence measure~$\mathscr D$ in this context. A set of axioms for dependence measures was proposed by R\'enyi~\cite{R59}, but other approaches and generalizations have been discussed as well, see e.\,g.~\cite{SW81,S84,KS89,SS07}. For the present purpose we find that a dependence measure which does not satisfy all of R\'enyi's axioms can yield a good importance measure. On the other hand, from R\'enyi's axioms alone certain desirable properties of importance measures cannot be proved. For example, R\'enyi's approach is non-ordinal, i.\,e.\ the axioms do not facilitate the comparison of different degrees of dependence, hence one cannot establish properties that compare relative importances, such as that the weakest component in a series system should have the largest importance. Specifically, we shall consider covariance $\mathscr D(U,V)=\cov(U,V)$ in the present paper. As is well known, covariance only detects linear dependence and is not nonnegative and normalized, hence not all of R\'enyi's axioms are satisfied. Still we can show that covariance can lead to a reasonable importance measure. If the system under consideration is coherent the variables~$X_i$ and~$X$ in the binary case and~$T_i$ and~$T$ in the continuous time case are associated in the sense of Esary, Proschan and Walkup~\cite{EPW67,BP75}, and in this case $\cov(X_i,X)\geq0$ or $\cov(T_i,T)\geq0$, and moreover uncorrelatedness implies independence. Thus, in fact, on the set of pairs of associated random variables, covariance is a proper dependence measure, and we can use it in this case. Nevertheless, for these reasons we do not define the term ``dependence measure'' in a formal way but use it in a rather loose fashion.\par
In the following we shall consider covariance and mutual information in place of the dependence measure~$\mathscr D$. We find that in the binary case covariance leads to a reasonable importance measure. In the continuous time case there are several ways to employ covariance, we shall discuss the possibilities $\cov(T_i,T)$ and $\displaystyle\sup_{s,t\geq0}\bigl(\cov(\mathbbm1_{\{T_i>s\}},\mathbbm1_{\{T>t\}})\bigr)$ (here~$\mathbbm1_A$ denotes the indicator function of the set~$A$), the former being related to the $\mathrm L^1$-distance between the joint distribution functions of~$T_i$ and~$T$ and the product of its marginals, and the latter to the $\mathrm L^\infty$-distance. We find that only the latter construction leads to a reasonable importance measure. Moreover, in the binary case we shall also discuss mutual information in place of~$\mathscr D$, which is based on Shannon entropy, and in this way we complement and extend some results of~\cite{EJSS14}, and we correct an error. We do not provide a corresponding discussion of mutual information in the continuous time case due to the added technical complications arising from the fact that the joint distribution of~$T_i$ and~$T$ is not absolutely continuous, hence mutual information is not defined. See however the remarks in~\cite{EJSS14} as well as~\cite{EJS14}. In order to assess whether an importance measure is reasonable, we investigate if it is able to detect irrelevant components or complete dependence, and whether it assigns large importance to unreliable series components and reliable parallel components.\par
The paper is organized as follows. In Section~\ref{sec::BinarySystems} we discuss the application of covariance and mutual information in the case of binary systems, whereas Section~\ref{sec::ContinuousTime} treats the continuous time case. Section~\ref{sec::ConcludingRemarks} closes with some concluding remarks.

\section{Binary Systems}
\label{sec::BinarySystems}

\subsection{Preliminaries}
\label{sec::PreliminariesBinarySystem}
We introduce some terms and notation used throughout the paper, see~\cite{BP75} for background information on the theory of binary coherent systems. Consider a system of~$n$ components. The state of the components is given by the~$n$ binary random variables $X_1,\ldots,X_n\in\{0,1\}$, defined on some common probability space~$(\Omega,\mathscr F,\mathbb P)$\kern1pt; we assume the random variables to be independent throughout. We write $p_i=\mathbb P\{X_i=1\}$ for the component reliabilities, and $\bar p_i=1-p_i=\mathbb P\{X_i=0\}$, as well as $\boldsymbol p=(p_1,\ldots,p_n)$ and $\boldsymbol 1=(1,\ldots,1)$. We shall frequently use the familiar notation~$(\boldsymbol x,0_i)$ and~$(\boldsymbol x,1_i)$, where
\[(\boldsymbol x,y_i)=(x_1,\ldots,x_{i-1},y,x_{i+1},\ldots,x_n),\quad i=1,\ldots,n,\ y\in\{0,1\},\]
and with $\boldsymbol x=(x_1,\ldots,x_n)$. The \emph{structure function} $\phi:\{0,1\}^{\times n}\longrightarrow\{0,1\}$ of the system determines its functioning or failure as a function of the states of the components. We call a structure function \emph{monotone} if it is nondecreasing (i.\,e.\ if $\boldsymbol x=(x_1,\ldots,x_n)$ and $\boldsymbol y=(y_1,\ldots,y_n)$ are two binary vectors with $x_i\leq y_i$ for $i=1,\ldots,n$, which is written as $\boldsymbol x\leq\boldsymbol y$, then $\phi(\boldsymbol x)\leq\phi(\boldsymbol y)$). Component~$i$ is called \emph{irrelevant} if $\phi(\boldsymbol x,1_i)=\phi(\boldsymbol x,0_i)$ for all~$\boldsymbol x$. If~$\phi$ is both monotone and has no irrelevant components, it is called \emph{coherent}. Throughout we will write $X=\phi(X_1,\ldots,X_n)$ as an abbreviation. The reliability of the system is given by
\[h_\phi(\boldsymbol p)=\mathbb P\{\phi(X_1,\ldots,X_n)=1\}=\mathbb E(\phi(X_1,\ldots,X_n))=\mathbb E(X),\]
where~$\mathbb E$ is the expectation corresponding to~$\mathbb P$. By independence~$h_\phi$ is a function of~$\boldsymbol p$ alone, and in fact a polynomial of degree less or equal than one in each~$p_1,\ldots,p_n$, as can be seen by pivotal decomposition on component~$i$\kern1pt: 
\begin{equation}
\phi(\boldsymbol x)=x_i\phi(\boldsymbol x,1_i)+(1-x_i)\phi(\boldsymbol x,0_i). 
\end{equation}
We call~$h_\phi$ the \emph{reliability function} of~$\phi$. The following obvious fact will be needed repeatedly\kern1pt:
\begin{equation}
\mathbb P\{X=x\mid X_i=y\}=h_\phi(\boldsymbol p,y_i)^x(1-h_\phi(\boldsymbol p,y_i))^{1-x},
\label{eq::ConditionalReliabilityFunction}
\end{equation}
where $x,y\in\{0,1\}$ and $i=1,\ldots,n$.\par
In this framework the most popular importance measure is the \emph{Birnbaum measure}~\cite{B69}. For a coherent structure function~$\phi$ it is defined as
\begin{equation}
\begin{aligned}
I_\phi^{\rm B}(i)&=\mathbb P\{\phi(\boldsymbol X,1_i)-\phi(\boldsymbol X,0_i)=1\}=h_\phi(\boldsymbol p,1_i)-h_\phi(\boldsymbol p,0_i)\\
&=\frac{\partial h_\phi}{\partial p_i}(\boldsymbol p),
\end{aligned}
\label{eq::BirnbaumMeasure}
\end{equation}
with the random vector $\boldsymbol X=(X_1,\ldots,X_n)$. Notice that~\eqref{eq::BirnbaumMeasure} does not depend on~$p_i$\kern1pt; however, the ordering of the components induced by~$I_\phi^{\rm B}$ does, in general, depend on all~$p_1,\ldots,p_n$.

\subsection{Covariance}
\label{sec::ImportanceCovarianceBinary}
The underlying idea of using dependence measures to construct importance indices is that a component should be important if its failure is strongly correlated with system failure. Moreover, for a repairable system, we could add that a component should be important when its repair is strongly correlated with system repair~\cite{NG09}. Thus, as explained in Section~\ref{sec::Introduction}, if~$\mathscr D(\cdot,\cdot)$ denotes some dependence measure defined at least on the set of binary variables, we can define a notion of importance of component~$i$ as $I_\phi^{\mathscr D}(i)=\mathscr D(X_i,X)$. For a sensible importance measure we should require $I_\phi^{\mathscr D}\geq0$. Generally, what matters is not the absolute values of the component importances, but in the first place the ordering of the components induced by a particular measure and secondly the relative differences between importances. To facilitate comparison between components it is sometimes advantageous to introduce a normalized version of~$I_\phi^{\mathscr D}$ by
\begin{equation}
\widehat I_\phi^{\mathscr D}(i)=I_\phi^{\mathscr D}(i)\Big/\sum_{j=1}^nI_\phi^{\mathscr D}(j),\quad i=1,\ldots,n,
\end{equation}
with the property that $\widehat I_\phi^{\mathscr D}(1)+\cdots+\widehat I_\phi^{\mathscr D}(n)=1$. Notice that the orderings of the components induced by~$I_\phi^{\mathscr D}$ and~$\widehat I_\phi^{\mathscr D}$ are the same.\par
Consider two random variables~$U$ and~$V$ and recall that their covariance is defined as $\cov(U,V)=\mathbb E(UV)-\mathbb E(U)\mathbb E(V)$. Following the above idea we define the \emph{covariance importance} of component~$i$ by
\begin{equation}
I_\phi^{\rm c}(i)=\cov(X_i,X)=\cov(X_i,\phi(X_1,\ldots,X_n)),\quad i=1,\ldots,n,
\label{eq::CovarianceImp}
\end{equation}
for a not necessarily coherent structure function~$\phi$. Recall our remarks in Section~\ref{sec::Introduction} that covariance, in general, is not a proper dependence measure. However, in this particular case it will be used on the set of pairs variables which are associated in the sense of Esary, Proschan and Walkup~\cite{EPW67,BP75}, and we will show that~\eqref{eq::CovarianceImp} indeed defines a valid importance measure, cf.~\thref{prop::CovarianceImportance} below.
\par
We start by establishing the following representation of~$I_\phi^{\rm c}$ in terms of the reliability function, using~\eqref{eq::ConditionalReliabilityFunction} and~\eqref{eq::CovarianceImp}\kern1pt:
\begin{equation}
\begin{aligned}
I_\phi^{\rm c}(i)&=\mathbb P\{X=1\mid X_i=1\}\mathbb P\{X_i=1\}-\mathbb P\{X=1\}\mathbb P\{X_i=1\}\\
&=p_i(h_\phi(\boldsymbol p,1_i)-h_\phi(\boldsymbol p)).
\end{aligned}
\label{eq::CovarianceImportanceReliabFunct}
\end{equation}
Notice that, unlike the Birnbaum importance~\eqref{eq::BirnbaumMeasure}, the covariance importance depends on~$p_i$. By pivotal decomposition we find
\begin{equation}
\begin{aligned}
I_\phi^{\rm c}(i)&=p_i(h_\phi(\boldsymbol p,1_i)-h_\phi(\boldsymbol p))=\bar p_i(h_\phi(\boldsymbol p)-h_\phi(\boldsymbol p,0_i))\\
&=p_i\bar p_i(h_\phi(\boldsymbol p,1_i)-h_\phi(\boldsymbol p,0_i)).
\end{aligned}
\label{eq::CovarianceImportanceRelationReliabfunct}
\end{equation}
The quantities appearing in the parentheses in the first line of the last equation are well known importance measures, $I_\phi^{\rm ra}(i)=h_\phi(\boldsymbol p)-h_\phi(\boldsymbol p,0_i)$ is called the \emph{risk achievement,} and $I_\phi^{\rm rr}(i)=h_\phi(\boldsymbol p,1_i)-h_\phi(\boldsymbol p)$ is called the \emph{risk reduction} or \emph{improvement potential}~\cite{BS01}. Thus we have the relations
\begin{equation}
I_\phi^{\rm c}(i)=\bar p_iI_\phi^{\rm ra}(i)=p_iI_\phi^{\rm rr}(i)=p_i\bar p_iI_\phi^{\rm B}(i).
\label{eq::CovarianceBirnbaumRelation}
\end{equation}
\par
By definition~$I_\phi^{\rm c}(i)$ measures the amount of linear correlation between the events of system failure and failure of component~$i$. Thus it concerns a different aspect than the Birnbaum measure, which identifies the component that has the highest probability of being critical for system operability. According to~\eqref{eq::BirnbaumMeasure} this is equivalent to identifying the component for which varying its reliability~$p_i$ yields the largest impact on system reliability. Thus for identifying the components whose failure has the strongest association with system failure, covariance importance should be used. In a reliability context this may be applied to test the balancedness of a safety concept. In contrast, the Birnbaum measure is preferred for identifying the components to which expenditures should be directed when the overall system reliability is to be improved. Equation~\eqref{eq::CovarianceBirnbaumRelation} provides another perspective on covariance importance. For example, if two components have the same risk achievement in the system, then the more unreliable one will have higher covariance importance (and correspondingly for risk reduction). See~\cite{CPS98,BS01} for of the interpretation of some of the classical time-dependent importance measures in nuclear applications.\par
We are now ready to derive some basic properties of the covariance importance which corroborate that it defines a sensible importance measure.
\begin{proposition}
\label{prop::CovarianceImportance}
For the covariance importance the following three assertions hold\kern1pt:
\begin{enumerate}
\item If~$\phi$ is coherent, $0\leq I_\phi^{\rm c}(i)\leq\tfrac14$ for all $i=1,\ldots,n$. Moreover, if $I_\phi^{\rm c}(i)=\tfrac14$ then $X=X_i$ a.\,s.,\ and if $I_\phi^{\rm c}(i)=0$ then~$X_i$ and~$X$ are independent. If $I_\phi^{\rm c}(i)=0$ for all $\boldsymbol p=(p_1,\ldots,p_n)$ then component~$i$ is irrelevant.
\item As a function of~$\boldsymbol p$ the importance measure~$I_\phi^{\rm c}(i)$ is a polynomial~$Q_i$, i.\,e.\ $I_\phi^{\rm c}(i)=Q_i(\boldsymbol p)$ for $i=1,\ldots,n$.
\item If~$\phi'(\boldsymbol x)=1-\phi(\boldsymbol1-\boldsymbol x)$ denotes the dual structure function, we have $I_{\phi'}^{\rm c}(i)=Q_i(\boldsymbol 1-\boldsymbol p)$, with $Q_i(\boldsymbol p)=I_\phi^{\rm c}(i)$ as in~(ii).
\end{enumerate}
\end{proposition}
\begin{proof}
(i)\enspace For a coherent~$\phi$, the reliability function~$h_\phi$ is monotonic, hence the first inequality follows from~\eqref{eq::CovarianceImportanceReliabFunct}. By the Cauchy--Schwarz inequality we have
\begin{equation}
I_\phi^{\rm c}(i)=\cov(X_i,X)\leq\sqrt{\var(X_i)\var(X)}=\sqrt{h_\phi(\boldsymbol p)\bar h_\phi(\boldsymbol p)\cdot p_i\bar p_i}\leq\tfrac14.
\label{eq::CovarianceImportanceBound}
\end{equation}
Now suppose that $I_\phi^{\rm c}(i)=\tfrac14$. Then we have equality in~\eqref{eq::CovarianceImportanceBound}, which means that the square root attains its maximum, which is at $h_\phi(\boldsymbol p)=p_i=\tfrac12$. Thus it follows that $\var(X)=\var(X_i)=\tfrac14$. In this case it is well known that 
\[\frac{X-\mathbb E(X)}{\sqrt{\var(X)}}=\frac{X_i-\mathbb E(X_i)}{\sqrt{\var(X_i)}}\qquad\text{a.\,s.},\]
thus $X=X_i$ a.\,s. The next assertion is obvious from~\eqref{eq::CovarianceImportanceReliabFunct} and the properties of the reliability function. Finally, if $I_\phi^{\rm c}(i)=0$ for all~$\boldsymbol p$, then from~\eqref{eq::CovarianceImportanceRelationReliabfunct} we have $h_\phi(\boldsymbol p,1_i)=h_\phi(\boldsymbol p,0_i)$ for all~$\boldsymbol p$, which implies irrelevance of component~$i$. Statement~(ii)\ is obvious from~\eqref{eq::CovarianceImportanceReliabFunct}, and concerning~(iii)\ we notice that
\[I_{\phi'}^{\rm c}(i)=\cov(X_i,1-\phi(\boldsymbol1-\boldsymbol X))=\cov(1-X_i,\phi(\boldsymbol 1-\boldsymbol X))=Q_i(\boldsymbol1-\boldsymbol p),\]
as was to be shown.
\end{proof}
\begin{remark}
For a cohernet~$\phi$, since $I_\phi^{\rm c}(i)=\cov(X_i,X)\geq0$ and~$X_i$ and~$X$ are binary, they are associated. Association of~$X$ and~$X_i$ can also be proved directly\kern1pt: Since~$X_1,\ldots,X_n$ are independent, they are associated, so the claim follows since~$\phi$ and the projection $\pr_i(x_1,\ldots,x_n)=x_i$ are nondecreasing. Then the last statement in claim~(i) of \thref{prop::CovarianceImportance}, i.\,e.\  independence of~$X_i$ and~$X$ if~$I_\phi^{\rm c}(i)=0$, follows from $\cov(X_i,X)=0$ since the variables are associated.
\end{remark}
\par
Notice that we have an inequality between~$I_\phi^{\rm B}$ and~$I_\phi^{\rm c}$, i.\,e.\ $I_\phi^{\rm B}(i)\geq I_\phi^{\rm c}(i)$ for all $i=1,\ldots,n$ if~$\phi$ is coherent\kern1pt; this follows immediately from~\eqref{eq::CovarianceBirnbaumRelation}. We now calculate covariance importance in some explicit examples.
\begin{example}
Consider a series system (i.\,e.\ $\phi(\boldsymbol x)=x_1\cdots x_n$), then we obtain the covariance importance from~\eqref{eq::CovarianceImportanceRelationReliabfunct} as
\[I_\phi^{\rm c}(i)=\bar p_ih_\phi(\boldsymbol p)=\bar p_ip_ih_\phi(\boldsymbol p,1_i),\quad\widehat I_\phi^{\rm c}(i)=\frac{\bar p_i}{\bar p_1+\cdots+\bar p_n}.\]
 For a parallel system we obtain from \thref{prop::CovarianceImportance} the corresponding covariance importances by replacing~$\bar p_i$ by~$p_i$ for $i=1,\ldots,n$ in the above equation.
\end{example}
\begin{example}
Consider a $k$-out-of-$n:{\rm G}$ system with independent components, which is functioning if of its~$n$ components at least~$k$ are working. Its structure function is given by 
\begin{equation}
\phi(\boldsymbol x)=\sum_{\substack{\boldsymbol\epsilon\in\{0,1\}^{\times n}\\ \epsilon_1+\cdots+\epsilon_n\geq k}}x_1^{\epsilon_1}\cdots x_n^{\epsilon_n}(1-x_1)^{1-\epsilon_1}\cdots(1-x_n)^{1-\epsilon_n},
\label{eq::koutofnStructureFunct}
\end{equation}
and reliability function reads
\begin{equation}
h_{n,k}(\boldsymbol p)=\sum_{\substack{\boldsymbol\epsilon\in\{0,1\}^{\times n}\\ \epsilon_1+\cdots+\epsilon_n\geq k}}p_1^{\epsilon_1}\cdots p_n^{\epsilon_n}\bar p_1^{1-\epsilon_1}\cdots\bar p_n^{1-\epsilon_n},
\label{eq::koutofnReliabilityFunct}
\end{equation}
where we denoted the reliability function by~$h_{n,k}(\boldsymbol p)$ to display the dependence on~$k$ and~$n$ explicitly. We obtain from~\eqref{eq::CovarianceImportanceReliabFunct} and~\eqref{eq::CovarianceImportanceRelationReliabfunct} for the covariance importance of the $k$-out-of-$n:{\rm G}$ system
\begin{align*}
I_\phi^{\rm c}(i)&=p_i(h_{n-1,k-1}(p_1,\ldots,p_{i-1},p_{i+1},\ldots,p_n)-h_{n,k}(\boldsymbol p))\\
&=\bar p_i(h_{n,k}(\boldsymbol p)-h_{n-1,k}(p_1,\ldots,p_{i-1},p_{i+1},\ldots,p_n)).
\end{align*}
\end{example}
From~\eqref{eq::koutofnStructureFunct} we can derive another representation of the covariance importance of the  $k$-out-of-$n:{\rm G}$ system\kern1pt:
\begin{multline*}
I_\phi^{\rm c}(i)=\cov(X_i,\phi(\boldsymbol X))\\
\begin{aligned}
&=\sum_{\substack{\boldsymbol\epsilon\in\{0,1\}^{\times n}\\ \epsilon_1+\cdots+\epsilon_n\geq k}}\cov(X_i,X_i^{\epsilon_i}(1-X_i)^{1-\epsilon_i})\cdot p_1^{\epsilon_1}\cdots\check p_i^{\epsilon_i}\cdots p_n^{\epsilon_n}\bar p_1^{1-\epsilon_1}\cdots\check{\bar p}_i^{1-\epsilon_i}\cdots\bar p_n^{1-\epsilon_n}\\
&=\sum_{\substack{\boldsymbol\epsilon\in\{0,1\}^{\times n}\\ \epsilon_1+\cdots+\epsilon_n\geq k}}(-1)^{1-\epsilon_i}(p_1^{\epsilon_1}\cdots p_{i-1}^{\epsilon_{i-1}}p_ip_{i+1}^{\epsilon_{i+1}}\cdots p_n^{\epsilon_n})\cdot{}\\
&\mkern300mu{}\cdot(\bar p_1^{1-\epsilon_1}\cdots\bar p_{i-1}^{1-\epsilon_{i-1}}\bar p_i\bar p_{i+1}^{1-\epsilon_{i+1}}\cdots\bar p_n^{1-\epsilon_n}),
\end{aligned}
\end{multline*}
where $p_1\cdots\check p_i\cdots p_n$ signifies that~$p_i$ is to be dropped from the product.\par
\begin{example}
In order to demonstrate that
\begin{table}[b]
\caption{Component ordering according to the Birnbaum and covariance importance in the 2-out-of-3 system for various component reliabilities. The last two columns contain the component indices in descending order of importance.}
\begin{center}
\begin{tabular}{lllll}
\multicolumn{3}{l}{\textit{Component reliabilities}} & \multicolumn{2}{l}{\textit{Induced ordering}}\\
$p_1$ & $p_2$ & $p_3$ & $I_\phi^{\rm B}$ & $I_\phi^{\rm c}$\\ \hline
0.1   & 0.2   & 0.3   & $3\enskip2\enskip1$ & $1\enskip2\enskip3$\\
0.3   & 0.4   & 0.5   & $3\enskip2\enskip1$ & $1\enskip3\enskip2$\\
0.5   & 0.6   & 0.7   & $1\enskip2\enskip3$ & $3\enskip1\enskip2$\\
0.7   & 0.8   & 0.9   & $1\enskip2\enskip3$ & $3\enskip2\enskip1$
\end{tabular}
\end{center}
\label{tab::ExampleOrderingCovariance}
\end{table}
the covariance importance induces a different component ordering than the Birnbaum importance we consider a 2-out-of-3 system, and we suppose without loss of generality that $p_1\leq p_2\leq p_3$. From~\eqref{eq::koutofnReliabilityFunct} we obtain for the reliability function $h_{3,2}(\boldsymbol p)=p_1p_2+p_1p_3+p_2p_3-2p_1p_2p_3$, and thus from~\eqref{eq::CovarianceBirnbaumRelation}
\begin{align*}
I_\phi^{\rm c}(1)&=p_1\bar p_1\cdot I_\phi^{\rm B}(1)=p_1\bar p_1[p_2+p_3-2p_2p_3],\\
I_\phi^{\rm c}(2)&=p_2\bar p_2\cdot I_\phi^{\rm B}(2)=p_2\bar p_2[p_1+p_3-2p_1p_3],\\
I_\phi^{\rm c}(3)&=p_3\bar p_3\cdot I_\phi^{\rm B}(3)=p_3\bar p_3[p_1+p_2-2p_1p_2].
\end{align*}
For the component ordering induced by the Birnbaum measure we have the following result\kern1pt: If $p_1\leq p_2\leq p_3\leq\tfrac12$ then $I_\phi^{\rm B}(3)\geq I_\phi^{\rm B}(2)\geq I_\phi^{\rm B}(1)$, and if $\tfrac12\leq p_1\leq p_2\leq p_3$ then $I_\phi^{\rm B}(3)\leq I_\phi^{\rm B}(2)\leq I_\phi^{\rm B}(1)$\kern1pt; this follows by explicit computation or from a general result~\cite{BP83,BP75} on Schur convexity of~$h_{n,k}$. This is to be compared by the ordering induced by~$I_\phi^{\rm c}$, which is recorded for some component reliabilities in Table~\ref{tab::ExampleOrderingCovariance}. It is seen that in this example the covariance importance tends to order the components in the opposite order than the Birnbaum importance.
\end{example}
\par
The next result demonstrates that covariance importance is a sensible importance measure. It shows that~$I_\phi^{\rm c}$ assigns large importance to unreliable series components and reliable parallel components. Moreover, it shows a compatibility property with a modular structure of the system.
\begin{theorem}
For a coherent structure function~$\phi$ the following two assertions hold true\kern1pt:
\begin{enumerate}
\item Suppose that component~$i$ is in series with a coherent module of the system, and that this module contains a component~$j$ such that $p_i\leq p_j$, where $j\neq i$. Then $I_\phi^{\rm c}(i)\geq I_\phi^{\rm c}(j)$.
\item Suppose that component~$i$ is parallel to a coherent module of the system, and that this module contains a component~$j$ such that $p_i\geq p_j$, where $j\neq i$. Then $I_\phi^{\rm c}(i)\geq I_\phi^{\rm c}(j)$.
\end{enumerate}
\label{thm::CovarianceSeriesParallelOrder}
\end{theorem}
\begin{proof}
(i)\enspace By an appropriate numbering of the components we can assume without loss of generality that $\phi(\boldsymbol x)=\psi(x_1,\ldots,x_{i-1},x_i\chi(x_{i+1},\ldots,x_n))$, i.\,e.\ we assume $j>i$, and with~$\psi$ and~$\chi$ coherent structure functions. By pivotal decomposition we have
\begin{multline*}
\phi(\boldsymbol x)=x_i\chi(x_{i+1},\ldots,x_n)\psi(x_1,\ldots,x_{i-1},1)+{}\\
{}+(1-x_i\chi(x_{i+1},\ldots,x_n))\psi(x_1,\ldots,x_{i-1},0),
\end{multline*}
and we obtain for any $k\geq i$
\begin{multline}
\begin{aligned}
I_\phi^{\rm c}(k)&=\cov(X_k,X)\\
&=\mathbb E(\psi(X_1,\ldots,X_{i-1},1))\cov(X_k,X_i\chi(X_{i+1},\ldots,X_n))+{}\\
&\qquad{}+\mathbb E(\psi(X_1,\ldots,X_{i-1},0))\cov(X_k,1-X_i\chi(X_{i+1},\ldots,X_n))\\
&=\mathbb E(\psi(X_1,\ldots,X_{i-1},1)-\psi(X_1,\ldots,X_{i-1},0))\cdot{}\\
&\qquad{}\cdot\cov(X_k,X_i\chi(X_{i+1},\ldots,X_n)).
\end{aligned}
\label{eq::CovOrderThmEq1}
\end{multline}
Since~$\psi$ is coherent the first factor in the last equation is nonnegative, hence $I_\phi^{\rm c}(i)\geq I_\phi^{\rm c}(j)$ is equivalent to
\begin{multline*}
\cov(X_i,X_i\chi(X_{i+1},\ldots,X_n))=\\
\begin{aligned}
&=\mathbb E(X_i\chi(X_{i+1},\ldots,X_n))-p_i\mathbb E(X_i\chi(X_{i+1},\ldots,X_n))\\
&\geq\cov(X_j,X_i\chi(X_{i+1},\ldots,X_n))\\
&=\mathbb E(X_jX_i\chi(X_{i+1},\ldots,X_n))-p_j\mathbb E(X_i\chi(X_{i+1},\ldots,X_n)),
\end{aligned}
\end{multline*}
which is in turn equivalent to
\[p_i\mathbb E((1-X_j)\chi(X_{i+1},\ldots,X_n))\geq(p_i-p_j)\mathbb E(X_i\chi(X_{i+1},\ldots,X_n)).\]
Now since $p_i-p_j\leq0$ and the expectations are nonnegative we see that this inequality is satisfied. (ii)\enspace Again without loss of generality we can write $\phi(\boldsymbol x)=\psi(x_1,\ldots,x_{i-1},1-(1-x_i)(1-\chi(x_{i+1},\ldots,x_n)))$. In the same way as before we can show that for any $k\geq i$
\begin{multline*}
I_\phi^{\rm c}(k)=\cov(X_k,X)=\mathbb E(\psi(X_1,\ldots,X_{i-1},1)-\psi(X_1,\ldots,X_{i-1},0))\cdot{}\\
{}\cdot\cov(X_k,1-(1-X_i)(1-\chi(X_{i+1},\ldots,X_n))).
\end{multline*}
As in~(i) we see that  $I_\phi^{\rm c}(i)\geq I_\phi^{\rm c}(j)$ is equivalent to 
\[(p_j-p_i)\mathbb E((1-X_i)(1-\chi(X_{i+1},\ldots,X_n)))\leq\bar p_i\mathbb E(X_j(1-\chi(X_{i+1},\ldots,X_n))).\]
Since $p_j-p_i\leq0$ the last inequality holds and we conclude.
\end{proof}
\begin{remark}
The assertions of \thref{thm::CovarianceSeriesParallelOrder} also hold for the Birnbaum measure, as can be seen as follows. For the series case suppose that component~$i$ is in series with the rest of the system, and without loss of generality that $i=1$, i.\,e.\ $\phi(\boldsymbol x)=x_1\psi(x_2,\ldots,x_n)$, and $p_i\leq p_j$ for some $j=2,\ldots,n$. Then
\begin{align*}
I_\phi^{\rm B}(1)&=\mathbb E(\psi(X_2,\ldots,X_n))\\
&=p_j\mathbb E(\psi(X_2,\ldots,1,\ldots,X_n))+\bar p_j\mathbb E(\psi(X_2,\ldots,0,\ldots X_n))
\intertext{and}
I_\phi^{\rm B}(j)&=p_1\mathbb E\bigl(\psi(X_2,\ldots,1,\ldots,X_n)-\psi(X_2,\ldots,0,\ldots,X_n)\bigr),
\end{align*}
from which we conclude $I_\phi^{\rm B}(1)\geq I_\phi^{\rm B}(j)$, as desired. Now this result can be generalized to the case that component~$i$ is in series to a coherent module containing component~$j$ by using the following simple fact\kern1pt: Suppose that $\phi(\boldsymbol x)=\psi(\chi(x_1,\ldots,x_k),x_{k+1},\ldots,x_n)$ for two arbitrary coherent structure functions~$\psi$ and~$\chi$. Moreover, assume that $I_\chi^{\rm B}(i)\geq I_\chi^{\rm B}(j)$ for some $i,j=1,\ldots,k$. Then since
\[I_\phi^{\rm B}(\ell)=\frac{\partial h_\phi}{\partial p_\ell}(\boldsymbol p)=\frac{\partial h_\psi}{\partial h_\chi}\frac{\partial h_\chi}{\partial p_\ell}(\boldsymbol p)=\frac{\partial h_\psi}{\partial h_\chi}I_\chi^{\rm B}(\ell)\quad\text{for all}\ \ell=1,\ldots,k\]
it follows that $I_\phi^{\rm B}(i)\geq I_\phi^{\rm B}(j)$. The parallel case can be verified in a similar manner.
\end{remark}

\subsection{Mutual Information}
Another dependence measure that can be used in place of covariance in~\eqref{eq::CovarianceImp} is mutual information~\cite{Sh48,As65}, see also~\cite{ESS10,ESS15} for an overview on the use of information measures in statistics and reliability. The construction of importance measures from mutual information was discussed in~\cite{EJSS14}. We try to generalize some results from that work, and we correct an error.\par 
If~$U$ and~$V$ are two (finite valued) random variables, the \emph{mutual information} of~$U$ and~$V$ is defined as $I(U\mid V)=H(U)-H(U\mid V)$, where~$H(U)$ is the Shannon entropy of~$U$ and~$H(U\mid V)$ the conditional entropy of~$U$ given~$V$. Recall that if~$U$ and~$V$ take on the values~$u_1,\ldots,u_n$ and~$v_1,\ldots,v_m$ with probabilities $p(u_i)=\mathbb P\{U=u_i\}$, $p(v_i)=\mathbb P\{V=v_i\}$, and $p(u_i,v_j)=\mathbb P\{U=u_i,V=v_j\}$, then
\begin{align}
H(U)&=-\sum_{i=1}^np(u_i)\log(p(u_i)),\label{eq::EntropyDefinition}\\
H(V\mid U)&=-\sum_{i=1}^n\sum_{j=1}^mp(u_i,v_j)\log(p(v_j\mid u_i)),
\label{eq::ConditionalEntropyDefinition}
\end{align}
where $p(v_j\mid u_i)=\mathbb P\{V=v_j\mid U=u_i\}$ and with the convention $0\log(0)=0$. If the logarithms are taken with respect to the base~$2$, as we will assume in the following, then~$I(U\mid V)$ is the difference of the average number of bits (or answers to ``yes''-or-``no'' questions) required to determine the result of one observation of~$U$ before and after the result of an observation of~$V$ is revealed. The mutual information satisfies $I(U\mid V)=I(V\mid U)$, and $I(U\mid V)\geq0$ with equality if and only if~$U$ and~$V$ are independent\kern1pt; moreover, $I(U\mid V)=H(U)+H(V)-H(U,V)$, where~$H(U,V)$ is the Shannon entropy corresponding to the joint distribution of~$U$ and~$V$.\par
Returning to the framework of Section~\ref{sec::PreliminariesBinarySystem}, we define the \emph{information importance} of component~$i$ by
\begin{equation}
I_\phi^{\rm inf}(i)=I(X_i\mid X)=I(X\mid X_i),\quad i=1,\ldots,n.
\label{eq::InformationImportanceDef}
\end{equation}
Using~\eqref{eq::ConditionalReliabilityFunction}, we can, in analogy to~\eqref{eq::CovarianceImportanceReliabFunct}, express it in terms of the reliability function as
\begin{align}
I_\phi^{\rm inf}(i)&=\sum_{x,y\in\{0,1\}}\mathbb P\{X=x\mid X_i=y\}\mathbb P\{X_i=y\}\cdot{}\nonumber\\
&\qquad{}\cdot\log\Bigl(\frac{\mathbb P\{X=x\mid X_i=y\}}{\mathbb P\{X=x\}}\Bigr)\nonumber\\
&=\sum_{x,y\in\{0,1\}}p_i^y\bar p_i^{1-y}h_\phi(\boldsymbol p,y_i)^x\bar h_\phi(\boldsymbol p,y_i)^{1-x}\cdot{}\nonumber\\
&\qquad{}\cdot\log\Bigl(\frac{h_\phi(\boldsymbol p,y_i)^x\bar h_\phi(\boldsymbol p,y_i)^{1-x}}{h_\phi(\boldsymbol p)^x\bar h_\phi(\boldsymbol p)^{1-x}}\Bigr),
\label{eq::InformationImportanceRep}
\end{align}
with $\bar h_\phi=1-h_\phi$.
\begin{proposition}
The following assertions hold\kern1pt:
\begin{enumerate}
\item For an arbitrary~$\phi$, $0\leq I_\phi^{\rm inf}(i)\leq1$ for all $i=1,\ldots,n$.
\item If $I_\phi^{\rm inf}(i)=0$ then~$X_i$ and~$X$ are independent.
\item If~$\phi$ is coherent and $I_\phi^{\rm inf}(i)=1$ then $X_i=X$ and $\mathbb P\{X_i=1\}=\tfrac12$.
\item If $F_i(\boldsymbol p)=I_\phi^{\rm inf}(i)$ denotes~$I_\phi^{\rm inf}(i)$ as a function of~$\boldsymbol p$, then for the dual structure function~$\phi'$ we have $I_{\phi'}^{\rm inf}(i)=F_i(\boldsymbol1-\boldsymbol p)$.
\end{enumerate}
\label{lem::InformationImportanceProp}
\end{proposition}
\begin{proof}
(i)\enspace From the definition of mutual information we have $I(X_i\mid X)\leq H(X_i)\leq1$, since~$X_i$ is a binary random variable. (ii)\enspace is clear. (iii)\enspace If $I_\phi^{\rm inf}(i)=H(X_i)-H(X_i\mid X)=1$ it follows that $H(X_i)=1$ and $H(X_i\mid X)=0$. The first equality implies $\mathbb P\{X_i=1\}=\tfrac12$ and the latter that $X_i=f(X)$ for some function $f:\{0,1\}\longrightarrow\{0,1\}$. Now since~$X_i$ and~$X$ are binary and~$\phi$ is coherent, the function~$f$ must be strictly monotonic, i.\,e.\ $X_i=X$. (iv)\enspace Since~$X$ and~$X_i$ are binary random variables, we have $H(X_i)=H(1-X_i)$ and $H(X_i\mid X)=H(1-X_i\mid 1-X)$ from~\eqref{eq::EntropyDefinition} and~\eqref{eq::ConditionalEntropyDefinition}. Thus is follows that
\begin{align*}
I_{\phi'}^{\rm inf}(i)&=I(X_i\mid1-\phi(\boldsymbol 1-\boldsymbol X))=H(X_i)-H(X_i\mid1-\phi(\boldsymbol1-\boldsymbol X))\\
&=H(1-X_i)-H(1-X_i\mid\phi(\boldsymbol1-\boldsymbol X))=I(1-X_i\mid\phi(\boldsymbol1-\boldsymbol X))\\
&=F_i(\boldsymbol1-\boldsymbol p),
\end{align*}
by using~\eqref{eq::EntropyDefinition} and~\eqref{eq::ConditionalEntropyDefinition} again.
\end{proof}
We note that a component has information importance~$H(X_i)$ if $X_i=X$, and information importance~$1$ if in addition $p_i=\tfrac12$.\par
\begin{example}
Consider a series system (i.\,e.\ $\phi(\boldsymbol x)=x_1\cdots x_n$). We evaluate~\eqref{eq::InformationImportanceRep} explicitly. In this case we have $h_\phi(\boldsymbol p,0_i)=0$ and $\bar h_\phi(\boldsymbol p,0_i)=1$. By an explicit calculation,
\begin{equation}
I_\phi^{\rm inf}(i)=\Bigl(p_i\coprod_{j\neq i}\bar p_j\Bigr)\log\Bigl(\coprod_{j\neq i}\bar p_j\Big/\coprod_{j=1}^n\bar p_j\Bigr)-\Bigl(\prod_{j=1}^np_j\Bigr)\log(p_i)-\bar p_i\log\Bigl(\coprod_{j=1}^n\bar p_j\Bigr),
\label{eq::InformationImportanceSeries}
\end{equation}
where $\coprod_jp_j=1-\prod_j(1-p_j)$. From \thref{lem::InformationImportanceProp} we see that the information importance of component~$i$ of a parallel system is obtained by replacing~$p_j$ by~$\bar p_j$ for $j=1,\ldots,n$ in~\eqref{eq::InformationImportanceSeries}.
\end{example}
\par
Next we look for a result analogous to \thref{thm::CovarianceSeriesParallelOrder}. However, as it turns out we could only obtain weaker results. We write $H(p)=-p\log(p)-\bar p\log(\bar p)$ for the entropy of a binary random variable~$U$ with $\mathbb P\{U=1\}=p$. Recall that the binary entropy function $p\mapsto H(p)$ is monotonically increasing on~$[0,\tfrac12]$, and decreasing on~$[\tfrac12,1]$.
\begin{theorem}
For an arbitrary structure function~$\phi$ the following three assertions hold true\kern1pt:
\begin{enumerate}
\item For $i,j\in\{1,\ldots,n\}$ we have $I_\phi^{\rm inf}(i)\geq I_\phi^{\rm inf}(j)$ if and only if
\begin{equation}
p_jH(h_\phi(\boldsymbol p,1_j))+\bar p_jH(h_\phi(\boldsymbol p,0_j))\geq p_iH(h_\phi(\boldsymbol p,1_i))+\bar p_iH(h_\phi(\boldsymbol p,0_i)).
\label{eq::InfImpComparison}
\end{equation}
\item Suppose that component~$i$ is in series with the rest of the system, that $p_i\leq p_j$ for some $j\neq i$, and that $\tfrac12\geq h_\phi(\boldsymbol p,1_j)\geq h_\phi(\boldsymbol p,1_i)$ or $\tfrac12\leq h_\phi(\boldsymbol p,1_j)\leq h_\phi(\boldsymbol p,1_i)$. Then $I_\phi^{\rm inf}(i)\geq I_\phi^{\rm inf}(j)$.
\item Suppose that component~$i$ is parallel to the rest of the system, that $p_i\geq p_j$ for some $j\neq i$, and that $\tfrac12\geq h_\phi(\boldsymbol p,0_j)\geq h_\phi(\boldsymbol p,0_i)$ or $\tfrac12\leq h_\phi(\boldsymbol p,0_j)\leq h_\phi(\boldsymbol p,0_i)$. Then $I_\phi^{\rm inf}(i)\geq I_\phi^{\rm inf}(j)$. 
\end{enumerate}
\label{thm::InformationImportanceOderdingGen}
\end{theorem}
\begin{proof}
(i)\enspace According to~\eqref{eq::InformationImportanceDef}, $I_\phi^{\rm inf}(i)\geq I_\phi^{\rm inf}(j)$ is equivalent to $H(X\mid X_j)\geq H(X\mid X_i)$, which in turn may be expressed as~\eqref{eq::InfImpComparison}, using~\eqref{eq::EntropyDefinition}, \eqref{eq::ConditionalEntropyDefinition} and~\eqref{eq::ConditionalReliabilityFunction}. (ii)\enspace If~$i$ is in series to the rest of the system, $h_\phi(\boldsymbol p,0_i)=0$, hence $H(h_\phi(\boldsymbol p,0_i))=0$. Moreover, the assumptions and the monotonicity properties of $p\mapsto H(p)$ imply $H(h_\phi(\boldsymbol p,1_j))\geq H(h_\phi(\boldsymbol p,1_i))$, thus $p_jH(h_\phi(\boldsymbol p,1_j))\geq p_iH(h_\phi(\boldsymbol p,1_i))$, and we see that~\eqref{eq::InfImpComparison} is satisfied. (iii)\enspace If~$i$ is parallel to the rest of the system, $h_\phi(\boldsymbol p,1_i)=1$, hence $H(h_\phi(\boldsymbol p,1_i))=0$. The rest follows as in~(ii) from the monotonicity properties of the binary entropy function and from~\eqref{eq::InfImpComparison}.
\end{proof}
For purely parallel or series systems assertions~(ii) and~(iii) of \thref{thm::InformationImportanceOderdingGen} can be considerably sharpened. To this end we need the following result.
\begin{lemma}
For any $\alpha\in[0,1]$ the function $p\mapsto H(\alpha p)/p$ is decreasing on~$\mathopen]0,1]$.
\label{lem::BinaryEntropyMonotonicity}
\end{lemma}
\begin{proof}
Recall the definition of the logit function $\logit(p)=\displaystyle\log\Bigl(\frac{p}{1-p}\Bigr)$, and that for the derivative of the binary entropy function $H'(p)=-\logit(p)$. Now the inequality
\[\frac{\de}{\de p}\frac{H(\alpha p)}{p}=-\frac{\alpha\logit(\alpha p)}{p}-\frac{H(\alpha p)}{p^2}\leq0\]
holds for any $\alpha,p\in\mathopen]0,1\mathclose[$ if and only if $-\logit(p)\leq H(p)/p$ for all $p\in\mathopen]0,1\mathclose[$. This latter inequality is easily rearranged to $\log(\bar p)\leq0$. Thus we conclude that the function in the statement of the lemma is decreasing.
\end{proof}
\begin{corollary}
Let $i,j=1,\ldots,n$ be two components, $i\neq j$. Then the following assertions hold\kern1pt:
\begin{enumerate}
\item For a series system and if $p_i\leq p_j$ it follows that $I_\phi^{\rm inf}(i)\geq I_\phi^{\rm inf}(j)$.
\item For a parallel system and if $p_i\geq p_j$ it follows that $I_\phi^{\rm inf}(i)\geq I_\phi^{\rm inf}(j)$.
\end{enumerate}
\label{cor::InformationImportanceSeriesPara}
\end{corollary}
\begin{proof}
(i)\enspace For the series system we have $h_\phi(\boldsymbol p)=p_1\cdots p_n$. Let us define the product $\displaystyle\alpha=\prod_{k\neq i,j}p_k\in[0,1]$. Then $h_\phi(\boldsymbol p,1_j)=\alpha p_i$ and $h_\phi(\boldsymbol p,1_i)=\alpha p_j$, and the assumption and \thref{lem::BinaryEntropyMonotonicity} yield
\[\frac{H(h_\phi(\boldsymbol p,1_j))}{p_i}=\frac{H(\alpha p_i)}{p_i}\geq\frac{H(\alpha p_j)}{p_j}=\frac{H(h_\phi(\boldsymbol p,1_i))}{p_j}.\]
Now since $h_\phi(\boldsymbol p,0_k)=0$ for all $k=1,\ldots,n$ it follows that inequality~\eqref{eq::InfImpComparison} is satisfied and we conclude by \thref{thm::InformationImportanceOderdingGen}. (ii)\enspace For the parallel system we have $h_\phi(\boldsymbol p)=1-\bar p_1\cdots\bar p_n$, and we define again a product $\displaystyle\alpha=\prod_{k\neq i,j}\bar p_k$. Then $\bar h_\phi(\boldsymbol p,0_j)=\alpha\bar p_i$ and $\bar h_\phi(\boldsymbol p,0_i)=\alpha\bar p_j$, and \thref{lem::BinaryEntropyMonotonicity} yields
\[\frac{H(\bar h_\phi(\boldsymbol p,0_j))}{\bar p_i}=\frac{H(\alpha\bar p_i)}{\bar p_i}\geq\frac{H(\alpha\bar p_j)}{\bar p_j}=\frac{H(\bar h_\phi(\boldsymbol p,1_i))}{\bar p_j},\]
where we used $H(p)=H(\bar p)$. Since $h_\phi(\boldsymbol p,1_k)=0$ for all $k=1,\ldots,n$ it again follows that inequality~\eqref{eq::InfImpComparison} is satisfied.
\end{proof}
It is worth noting that the results of \thref{cor::InformationImportanceSeriesPara} are included in~\cite{EJSS14}, however, the proof there contains some gaps. In fact, in the series case instead of $p_jH(h_\phi(\boldsymbol p,1_j))\geq p_iH(h_\phi(\boldsymbol p,1_i))$ as required by assertion~(i) of \thref{thm::InformationImportanceOderdingGen}, the method used in~\cite{EJSS14} only shows $p_jH(h_\phi(\boldsymbol p,1_j))\geq p_iH(h_\phi(\boldsymbol p,1_j))$\kern1pt; a similar gap exists for the parallel case.\par
Some numerical experiments with various reliability functions indicated that the conclusions of \thref{thm::CovarianceSeriesParallelOrder} also hold for information importance, i.\,e.,\ the additional assumptions in \thref{thm::InformationImportanceOderdingGen} are unnecessary. A proof appears to be difficult, so for now we can only formulate the following.
\begin{conjecture}
Let~$\phi$ be an arbitrary reliability function. Then the conclusions of \thref{thm::CovarianceSeriesParallelOrder} hold for~$I_\phi^{\rm inf}$.
\end{conjecture}

\section{Continuous Time}
\label{sec::ContinuousTime}
\subsection{Preliminaries}
\label{sec::PreliminariesContinuousTime}
In addition to the notation introduced in Section~\ref{sec::PreliminariesBinarySystem} we will need the following concepts. Again we consider a system of~$n$ components, but this time with random lifetimes~$T_1,\ldots,T_n$, where the~$T_i$ are nonnegative random variables. In the sequel we assume that they are independent and of finite variance. The system is defined by a structure function~$\phi$ as in Section~\ref{sec::PreliminariesBinarySystem}, and the system life time is denoted by~$T$. If~$K_1,\ldots,K_k$ and~$P_1,\ldots,P_p$ are the minimal cut sets and minimal path sets of a coherent~$\phi$, then $\phi(\boldsymbol x)=\displaystyle\max_{j=1,\ldots,p}\min_{i\in P_j}x_i=\min_{j=1,\ldots,k}\max_{i\in K_j}x_i$. A subset $S\subseteq\{1,\ldots,n\}$ is called a formation if it is a union of minimal path sets, and a collection~$P_{i_1},\ldots,P_{i_\ell}$ of minimal path sets such that $S=P_{i_1}\cup\cdots\cup P_{i_\ell}$ is called a representation of~$F$. A representation is called odd if~$\ell$ is odd, and even if~$\ell$ is even. If a formation~$S$ has~$n_{\rm o}$ odd and~$n_{\rm e}$ even representations, then its signed domination is defined as $b_\phi(F)=n_{\rm o}-n_{\rm e}$. Recall that~$\phi$ can be written as
\begin{equation}
\phi(\boldsymbol x)=\sum_{S\subseteq\{1,\ldots,n\}}b_\phi(S)\prod_{i\in S}x_i,
\label{eq::StructureFunctMultilinear}
\end{equation}
cf.~\cite{BT12}. If~$\tau_\phi(\boldsymbol t)$ is the system life time as a function of the component life times $\boldsymbol t=(t_1,\ldots,t_n)$, we have
\begin{equation}
\tau_\phi(\boldsymbol t)=\max_{j=1,\ldots,p}\min_{i\in P_j}t_i=\min_{j=1,\ldots,k}\max_{i\in K_j}t_i,
\label{eq::SystemLifetimeCutPathSets}
\end{equation}
and hence
\begin{equation}
T=\tau_\phi(T_1,\ldots,T_n).
\end{equation}
We write $X_i(t)=\mathbbm1_{\{T_i>t\}}$, $i=1,\ldots,n$, and $X(t)=\mathbbm1_{\{T>t\}}$ for the now time-dependent binary component and system state variables. Moreover, for any binary variable~$X$ we define $\bar X=1-X$. If the random variable~$S$ is less than~$T$ in the usual stochastic order, i.\,e.\ if $\mathbb P\{S>t\}\leq\mathbb P\{T>t\}$ for all $t\in\mathbb R$, we write~$S\preceq T$. Finally, for $s,t\in\mathbb R$ we write~$s\land t$ for~$\min\{s,t\}$, and~$s\lor t$ for~$\max\{s,t\}$.\par
For any component $i=1,\ldots,n$ we denote the joint distribution of~$T_i$ and~$T$ by
\[H_i(s,t)=\mathbb P\{T_i\leq s,T\leq t\},\quad s,t\in\mathbb R,\]
and we write~$F_i$ and~$F$ for the cumulative distribution functions of~$T_i$ and~$T$, respectively. Clearly, the joint distribution~$H_i$~contains all information about the dependence between~$T_i$ and~$T$, so an importance measure based on the dependence between these two random variables should be constructed from~$H_i$ alone.

\subsection{Covariance}
In this section we will parallel the developments of Section~\ref{sec::ImportanceCovarianceBinary} and define an importance measure based on covariance. In fact, in the continuous time setting there are several possible ways to introduce importance measures based on covariance. We start by defining the \emph{$\mathrm L^1$-covariance importance} of component~$i$ in the continuous time case by
\begin{equation}
I_\phi^{\rm c1}(i)=\cov(T_i,T)=\cov(T_i,\tau_\phi(T_1,\ldots,T_n)),\quad i=1,\ldots,n.
\label{eq::CovarianceImportanceCont}
\end{equation}
Whereas the covariance importance~\eqref{eq::CovarianceImp} may be regarded as an importance at a fixed time in the present setting, i.\,e.\ $I_\phi^{\rm c}(i,t)=\cov(X_i(t),X(t))$ for some fixed~$t\geq0$, the importance measure~\eqref{eq::CovarianceImportanceCont} depends on the distributions of the component life times~$T_1,\ldots,T_n$, or more precisely, on~$H_i$.\par
We start by establishing a representation of~\eqref{eq::CovarianceImportanceCont} in terms of binary variables, which also shows some connection to the binary case expounded in Section~\ref{sec::ImportanceCovarianceBinary}. We have
\begin{multline}
H_i(s,t)-F_i(s)F(t)\\
\begin{aligned}
&=\mathbb P\{\bar X_i(s)=1,\bar X(t)=1\}-\mathbb P\{\bar X_i(s)=1\}\mathbb P\{\bar X(t)=1\}\\
&=\mathbb E(\bar X_i(s)\bar X(t))-\mathbb E(\bar X_i(s))\mathbb E(\bar X(t))\\
&=\cov(\bar X_i(s),\bar X(t))\\
&=\cov(X_i(s),X(t))
\end{aligned}
\end{multline}
for any $s,t\geq0$. Now by Hoeffding's covariance lemma (see~\cite{L66} for a proof) we can write
\begin{equation}
\begin{aligned}
\cov(T_i,T)&=\int_0^\infty\int_0^\infty\bigl(H_i(s,t)-F_i(s)F(t)\bigr)\,\de s\de t\\
&=\int_0^\infty\int_0^\infty\cov(X_i(s),X(t))\,\de s\de t\\
&=\lVert H_i-F_iF\rVert_1.
\end{aligned}
\label{eq::L1CovarianceImportanbceHoeffding}
\end{equation}	
This shows that~$I_\phi^{\rm c1}(i)$ is actually the $\mathrm L^1$-distance between~$H_i$ and~$F_iF$. Notice that $\cov(X_i(s),X(t))\geq0$ for all $s,t\geq0$ since~$X_i(s)$ and~$X(t)$ are nondecreasing functions of~$T_1,\ldots,T_n$ in view of~\eqref{eq::SystemLifetimeCutPathSets}. The integrand~$\cov(X_i(s),X(t))$ in the last integral is, for $s=t$, equal to the correlation importance of component~$i$ for the system defined by~$\phi$ and the binary variables~$X_1(t),\ldots,X_n(t)$ for the component states.\par
In analogy to item~(i) of \thref{prop::CovarianceImportance} we have a corresponding but slightly weaker result for~$I_\phi^{\rm c1}$.
\begin{proposition}
If~$\phi$ is coherent, $0\leq I_\phi^{\rm c1}(i)\leq\sqrt{\var(T_i)\var(T)}$ for all $i=1,\ldots,n$. Moreover, if $I_\phi^{\rm c1}(i)=0$ then~$T_i$ and~$T$ are independent, and if $I_\phi^{\rm c1}(i)$ is maximal, i.\,e.\ $I_\phi^{\rm c1}(i)=\sqrt{\var(T_t)\var(T)}$, then~$T_i$ and~$T$ are~a.\,s linearly dependent.
\label{prop::L1CovarianceImportanceProperties}
\end{proposition}
\begin{proof}
Since~$T_1,\ldots,T_n$ are independent, they are associated, and moreover the projection~$\pr_i$ as well as~$\tau_\phi$ is a nondecreasing function. Hence the variables $T=\tau_\phi(T_1,\ldots,T_n)$ and~$T_i=\pr_i(T_1,\ldots,T_n)$ are associated as well, thus $\cov(T_i,T)\geq0$. The second inequality is the Cauchy--Schwarz inequality. Moreover, $I_\phi^{\rm c1}(i)=\cov(T_i,T)=0$ implies independence of~$T_i$ and~$T$ since they are associated. Finally, if~$I_\phi^{\rm c1}(i)$ is maximal, there is equality in the Cauchy--Schwarz inequality and thus~$T_i$ and~$T$ are~a.\,s.\ linearly dependent.
\end{proof}
\par
Instead of using the $\mathrm L^1$-distance between~$H_i$ and~$F_iF$, we can replace it by the $\mathrm L^\infty$-distance to define another importance measure based on covariance, which we call the \emph{$\mathrm L^\infty$-covariance importance\kern1pt:}
\begin{equation}
\begin{aligned}
I_\phi^{\rm c\infty}(i)&=\lVert H_i-F_iF\rVert_\infty=\sup_{s,t\geq0}\lvert H_i(s,t)-F_i(s)F(t)\rvert\\
&=\sup_{s,t\geq0}\lvert\cov(X_i(s),X(t))\rvert.
\end{aligned}
\label{eq::LinftyCovarianceImportanceDefinition}
\end{equation}
We remark that a dependence measure based on the $\mathrm L^\infty$-distance has been discussed in~\cite{SW81}. It is easy to see that we have the following version of \thref{prop::L1CovarianceImportanceProperties}.
\begin{proposition}
If~$\phi$ is coherent, $0\leq I_\phi^{\rm c\infty}(i)\leq\tfrac14$ for all $i=1,\ldots,n$. Moreover, if $I_\phi^{\rm c\infty}(i)=0$ then~$T_i$ and~$T$ are independent.
\end{proposition}
\begin{proof}
The proof follows trivially by using \thref{prop::CovarianceImportance} and~\eqref{eq::LinftyCovarianceImportanceDefinition}.
\end{proof}
We proceed by establishing a lemma which helps calculating the $\mathrm L^\infty$-covariance importance explicitly.
\begin{lemma}
Suppose that~$\phi$ is coherent. Then the $\mathrm L^\infty$-covariance importance can be written as
\begin{equation}
I_\phi^{\rm c\infty}(i)=\sup_{t\geq0}\bigl\{\cov(X_i(t),X(t))\bigr\}
\label{eq::LinftyCovariaceImpSup}
\end{equation}
for any $i=1,\ldots,n$.
\label{lem::LinftyCovarianceRep}
\end{lemma}
\begin{proof}
Fix $t\geq0$ and $i=1,\ldots,n$. From~\eqref{eq::LinftyCovarianceImportanceCalculation} below we have
\[\cov(X_i(s),X(t))=K(t)\bigl(\bar F_i(s\lor t)-\bar F_i(s)\bar F_i(t)\bigr),\quad s\geq0,\]
with~$K(t)\geq0$ since $\cov(X_i(s),X(t))\geq0$ and $\bar F_i(s\lor t)\geq\bar F_i(s)\bar F_i(t)$. Now
\[[0,\infty\mathclose[\ni s\mapsto\bar F_i(s\lor t)-\bar F_i(s)\bar F_i(t)=
\begin{cases}
F_i(s)\bar F_i(t)&:s\leq t\\
\bar F_i(s)F_i(t)&:s\geq t
\end{cases}\]
is maximal for~$s=t$. Thus we have established that $\displaystyle\sup_{s,t\geq0}\bigl(\cov(X_i(s),X(t))\bigr)\leq\sup_{t\geq0}\bigl(\cov(X_i(t),X(t))\bigr)$. Since the converse inequality holds trivially we have proved~\eqref{eq::LinftyCovariaceImpSup}.
\end{proof}
Using the representation~\eqref{eq::StructureFunctMultilinear} of the structure function we can write for any $i=1,\ldots,n$ and $s,t\geq0$
\begin{multline}
\cov(X_i(s),X(t))\\
\begin{aligned}
&=\sum_{S\subseteq\{1,\ldots,n\}}b_\phi(S)\cov\Bigl(X_i(s),\prod_{j\in S}X_j(t)\Bigr)\\
&=\sum_{S\subseteq\{1,\ldots,n\}\backslash\{i\}}b_\phi(S\cup\{i\})\mathbb E\Bigl(\prod_{j\in S}X_j(t)\Bigr)\cov(X_i(s),X_i(t))\\
&=\sum_{S\subseteq\{1,\ldots,n\}\backslash\{i\}}b_\phi(S\cup\{i\})\prod_{j\in S}\bar F_j(t)\bigl(\bar F_i(s\lor t)-\bar F_i(s)\bar F_i(t)\bigr),
\end{aligned}
\label{eq::LinftyCovarianceImportanceCalculation}
\end{multline}
so that by \thref{lem::LinftyCovarianceRep}
\begin{equation}
I_\phi^{\rm c\infty}(i)=\sup_{t\geq0}\left\{F_i(t)\sum_{S\subseteq\{1,\ldots,n\}\backslash\{i\}}b_\phi(S\cup\{i\})\prod_{j\in S\cup\{i\}}\bar F_j(t)\right\}.
\label{eq::LinftyCovarianceImportanceDistribFunct}
\end{equation}
This result lends itself to an explicit calculation of~$I_\phi^{\rm c\infty}$. A similar representation of~$I_\phi^{\rm c\infty}$ can be derived by using the cut set or path set representation~\eqref{eq::SystemLifetimeCutPathSets} of the reliability function instead of the signed domination form~\eqref{eq::StructureFunctMultilinear}.\par
We now consider dual structures. Let~$\phi'$ be the dual structure function of~$\phi$. As in \thref{prop::CovarianceImportance} we have
\[\cov(X_i(s),\phi'(X_1(t),\ldots,X_n(t)))=\cov(\bar X_i(s),\phi(\bar X_1(t),\ldots,\bar X_n(t))),\]
so by repeating the calculations in~\eqref{eq::LinftyCovarianceImportanceCalculation} we see that
\begin{equation}
\cov(\bar X_i(s),\bar X(t))=\sum_{S\subseteq\{1,\ldots,n\}\backslash\{i\}}b_\phi(S\cup\{i\})\prod_{j\in S\cup\{i\}}F_j(t)\bigl(F_i(s\land t)-F_i(s)F_i(t)\bigr),
\end{equation}
thus~$I_{\phi'}^{\rm c\infty}$ can be obtained from~\eqref{eq::LinftyCovarianceImportanceDistribFunct} by replacing~$F_i$ by~$\bar F_i$, $\bar F_j$ by~$F_j$, and retaining~$b_\phi$, i.\,e.\ 
\begin{equation}
I_{\phi'}^{\rm c\infty}(i)=\sup_{t\geq0}\left\{\bar F_i(t)\sum_{S\subseteq\{1,\ldots,n\}\backslash\{i\}}b_\phi(S\cup\{i\})\prod_{j\in S\cup\{i\}}F_j(t)\right\}.
\label{eq::LinftyCovarianceImportanceDistribFunctDual}
\end{equation}
Similarly, if we use~\eqref{eq::LinftyCovarianceImportanceCalculation} together with~\eqref{eq::L1CovarianceImportanbceHoeffding} to calculate~$I_\phi^{\rm c1}(i)$, then again a replacement of~$F_i$ by~$\bar F_i$, and~$\bar F_j$ by~$F_j$ yields~$I_{\phi'}^{\rm c1}(i)$.
\par
For the $\mathrm L^\infty$-covariance importance we can prove a result corresponding to \thref{thm::CovarianceSeriesParallelOrder}. As before this corroborates that~$I_\phi^{\rm c\infty}$ is a sensible importance measure in that it assigns large importances to unreliable series and reliable parallel components, and that it has some compatibility with a modular structure of the system.
\begin{theorem}
For a coherent structure function~$\phi$ the following two assertions hold true\kern1pt:
\begin{enumerate}
\item Suppose that component~$i$ is in series with a coherent module of the system, and that this module contains a component~$j$ such that that $T_i\preceq T_j$, where $j\neq i$. Then $I_\phi^{\rm c\infty}(i)\geq I_\phi^{\rm c\infty}(j)$.
\item Suppose that component~$i$ is parallel to a coherent module of the system, and that this module contains a component~$j$ such that $T_i\succeq T_j$, where $j\neq i$. Then $I_\phi^{\rm c\infty}(i)\geq I_\phi^{\rm c\infty}(j)$.
\end{enumerate}
\label{thm::CovarianceSeriesParallelOrderCont}
\end{theorem}
\begin{proof}
To verify~(i), we assume as in the proof of \thref{thm::CovarianceSeriesParallelOrder} and without loss of generality that the structure function of the system can be written as $\phi(\boldsymbol x)=\psi(x_1,\ldots,x_{i-1},x_i\chi(x_{i+1},\ldots,x_n))$, i.\,e.\ $j>i$, and with~$\psi$ and~$\chi$ coherent structure functions. In view of \thref{lem::LinftyCovarianceRep} it suffices to prove $\cov(X_i(t),X(t))\geq\cov(X_j(t),X(t))$ for all $t\geq0$. By using the assumption $T_i\preceq T_j$, which means $\mathbb E(X_i(t))\leq\mathbb E(X_j(t))$ for all $t\geq0$, this can be shown in exactly the same way as in the proof of \thref{thm::CovarianceSeriesParallelOrder} when~$X_i$ is replaced by~$X_i(t)$ and~$X$ by~$X(t)$, and~$p_i$, $p_j$ become~$\mathbb E(X_i(t))$, $\mathbb E(X_i(t))$, respectively. Similarly, by using \thref{lem::LinftyCovarianceRep}, it is seen that~(ii) can be checked in exactly the same way as in the proof of \thref{thm::CovarianceSeriesParallelOrder}.
\end{proof}
\begin{remark}
\label{rem::NatvigsMeasure}
It has been shown in~\cite{Nor86} that Natvig's importance measure~\cite{Nat79,Nat82, Nat85} can be written as $I_\phi^{\rm N}(i)=\cov(A_i(T_i),T)$, where~$\{A_i(t)\}_{t\geq0}$ is the compensator of the counting process $N_i(t)=\mathbbm1_{\{T_i\leq t\}}$. Thus~$I_\phi^{\rm N}$ it is closely related to~$I_\phi^{\rm c1}$. Moreover, a result similar to \thref{thm::CovarianceSeriesParallelOrderCont} for Natvig's measure has been demonstrated in~\cite{NG09}\kern1pt; this reference together with~\cite{Nat11} also consider the Natvig and Barlow--Proschan measures for repairable and for strongly coherent multistate systems.
\end{remark}
\par
The following example demonstrates that we cannot expect a result similar to \thref{thm::CovarianceSeriesParallelOrderCont} for~$I_\phi^{\rm c1}$.
\begin{example}
\label{ex::L1CovarianceImportanceTwoSeries}
Consider for the sake of simplicity two components in series. Then $\phi(x_1,x_2)=x_1x_2$ and for the signed domination we have $b(S)=1$ if $S=\{1,2\}$ and $b(S)=0$ for all other $S\subseteq\{0,1\}$. Thus from~\eqref{eq::LinftyCovarianceImportanceCalculation}
\begin{align*}
\cov(X_1(s),X(t))=(\bar F_1(s\lor t)-\bar F_1(s)\bar F_1(t))\bar F_2(t),\\
\cov(X_2(s),X(t))=(\bar F_2(s\lor t)-\bar F_2(s)\bar F_2(t))\bar F_1(t),
\end{align*}
and we suppose that the component life times are exponentially distributed, i.\,e.\ $\bar F_i(t)=\e^{-\lambda_it}$, $i=1,2$, and that $T_1\preceq T_2$, i.\,e.\ $\lambda_1\geq\lambda_2$. To calculate~$I_\phi^{\rm c1}$ we use the above identities together with~\eqref{eq::L1CovarianceImportanbceHoeffding} and split the area of integration in the double integral in two parts $\{(s,t):0\leq s\leq t\}$ and $\{(s,t):0\leq t\leq s\}$. This yields
\begin{align}
I_\phi^{\rm c1}(1)&=\int_0^\infty\int_0^t\cov(X_1(s),X(t))\,\de s\de t+\int_0^\infty\int_0^s\cov(X_1(s),X(t))\,\de t\de s\nonumber\\
&=\int_0^\infty\int_0^t(1-\e^{-\lambda_1s})\e^{-(\lambda_1+\lambda_2)t}\,\de s\de t+{}\nonumber\\
&\qquad{}+\int_0^\infty\int_0^s\e^{-\lambda_1s}(\e^{-\lambda_2t}-\e^{-(\lambda_1+\lambda_2)t})\,\de t\de s\nonumber\\
&=\frac{1}{(\lambda_1+\lambda_2)^2}-\frac{2}{\lambda_1(\lambda_1+\lambda_2)}-\frac{1}{\lambda_2(\lambda_1+\lambda_2)}+{}\nonumber\\
&\qquad{}+\frac{1}{\lambda_1(2\lambda_1+\lambda_2)}+\frac{1}{(\lambda_1+\lambda_2)(2\lambda_1+\lambda_2)}+\frac{1}{\lambda_1\lambda_2}.\label{eq::L1CovarianceImportanceExample}
\end{align}
To obtain~$I_\phi^{\rm c1}(2)$, $\lambda_1$~and~$\lambda_2$ have to be interchanged in the last equation. Now it can be seen after some algebra that we have $I_\phi^{\rm c1}(1)=I_\phi^{\rm c1}(2)$, that is, for two components in series with exponential life times the $\mathrm L^1$-covariance importance cannot distinguish between them. Moreover, this result can be generalized as can be seen from the structure of~\eqref{eq::LinftyCovarianceImportanceCalculation}\kern1pt: Consider a series system of~$n$ exponentially distributed components with failure rates~$\lambda_1,\ldots,\lambda_n$. Then~$I_\phi^{\rm c1}(1)$ is given by~\eqref{eq::L1CovarianceImportanceExample} when~$\lambda_2$ is replaced with $\lambda_2+\mu$, where $\mu=\lambda_3+\cdots+\lambda_n$, and~$I_\phi^{\rm c1}(2)$ when~$\lambda_1$ is replaced by~$\lambda_2$ and~$\lambda_2$ by~$\lambda_1+\mu$. Again some algebra shows that $I_\phi^{\rm c1}(1)=I_\phi^{\rm c1}(2)$, hence by symmetry for a series system with exponentially distributed component life times all components have the same $\mathrm L^1$-covariance importance. This clearly casts doubts on the utility of the $\mathrm L^1$-covariance importance as a good importance measure.\par
It is interesting to compare~$I_\phi^{\rm c1}$ with Natvig's importance measure~$I_\phi^{\rm N}$ that was mentioned in Remark~\ref{rem::NatvigsMeasure} in the context of this example. For the compensator of the counting process $N_i(t)=\mathbbm1_{\{T_i\leq t\}}$ we obtain $A_i(t)=\lambda_i(t\land T_i)$, thus $I_\phi^{\rm N}(i)=\cov(\lambda_iT_i,T)=\lambda_iI_\phi^{\rm c1}(i)$, thus Natvig's measure, in contrast to~$I_\phi^{\rm c1}$, distinguishes between the two components.
\end{example}
\begin{example}
\begin{figure}[t]
\centering
\includegraphics[width=1\columnwidth]{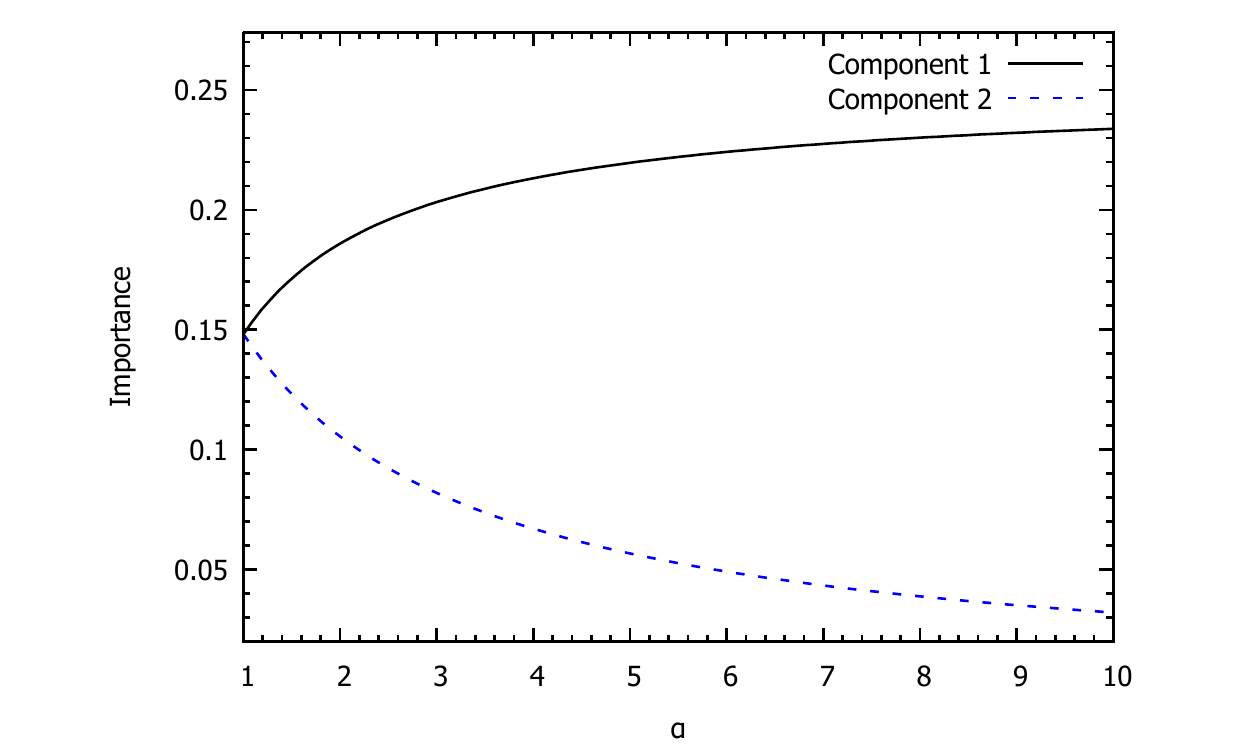}
\caption{$\mathrm L^\infty$-covariance importance for a two component series system with exponential component lifetimes, where $\lambda_1=\alpha\lambda_2$.}
\label{fig::LinftyExample}
\end{figure}
Consider again the system of two components in series with exponential life times from \thref{ex::L1CovarianceImportanceTwoSeries}. To calculate~$I_\phi^{\rm c\infty}(1)$ we use~\eqref{eq::LinftyCovarianceImportanceDistribFunct}, i.\,e.\ we need to find the maximum of $t\mapsto F_1(t)\bar F_1(t)\bar F_2(t)$ from the zero of its first derivative, which is at 
\[t_0=-\frac{1}{\lambda_1}\ln\Bigl(\frac{2\lambda_1+\lambda_2}{\lambda_1+\lambda_2}\Bigr),\]
hence we obtain
\begin{align*}
I_\phi^{\rm c\infty}(1)&=\Bigl(\frac{2\lambda_1+\lambda_2}{\lambda_1+\lambda_2}\Bigr)^{-\frac{\lambda_1+\lambda_2}{\lambda_1}}-\Bigl(\frac{2\lambda_1+\lambda_2}{\lambda_1+\lambda_2}\Bigr)^{-\frac{2\lambda_1+\lambda_2}{\lambda_1}},
\end{align*}
and for the other component by interchanging~$\lambda_1$ and~$\lambda_2$. If we put $\lambda_1=\alpha\lambda_2$ with $\alpha\geq1$ we get
\begin{align*}
I_\phi^{\rm c\infty}(1)&=\Bigl(\frac{2\alpha+1}{\alpha+1}\Bigr)^{-(\alpha+1)/\alpha}-\Bigl(\frac{2\alpha+1}{\alpha+1}\Bigr)^{-(2\alpha+1)/\alpha},\\
I_\phi^{\rm c\infty}(2)&=\Bigl(\frac{\alpha+2}{\alpha+1}\Bigr)^{-(\alpha+1)}-\Bigl(\frac{\alpha+2}{\alpha+1}\Bigr)^{-(\alpha+2)}.
\end{align*}
In Figure~\ref{fig::LinftyExample} the $\mathrm L^\infty$-covariance importances~$I_\phi^{\rm c\infty}(1)$ and~$I_\phi^{\rm c\infty}(2)$ of the two components are plotted as a function of~$\alpha$. Unlike the $\mathrm L^1$-covariance importance, the $\mathrm L^\infty$-covariance importance distinguishes between the two components and assigns a larger importance to the component whose lifetime distribution is smaller in the usual stochastic order, i.\,e.\ which has a larger failure rate.
\end{example}

\section{Concluding Remarks}
\label{sec::ConcludingRemarks}
We have studied component importance from the point of view of stochastic dependence, both in the binary (time-dependent) and continuous time case. In both cases we have used covariance as a dependence measure to construct importance measures, and we have been able to prove that the resulting measures order components in a natural way, i.\,e.\ assign large importance to unreliable series components and reliable parallel components. Moreover, in the binary case, similar to the developments in~\cite{EJSS14}, we have employed mutual information to obtain an importance measure, but we have only been able to obtain a weaker result on importance ordering for purely parallel and series systems\kern1pt; a generalization of this result remains a topic for further study. We have not addressed the use of mutual information to define a dependence measure in the continuous time case since this is associated with a number of additional technical complications. They are related to the joint distribution~$H_i$ of system life time~$T$ and component life time~$T_i$, $i=1,\ldots,n$, not being absolutely continuous for coherent systems (and not even absolutely continuous with respect to the product of the distributions of~$T_i$ and~$T$), hence mutual information is not defined in this case. A generalization of mutual information that covers this case has been introduced in~\cite{EJS14}, and some results concerning a corresponding importance measure are obtained in~\cite{EJSS14}. However, to prove an analog of \thref{thm::CovarianceSeriesParallelOrderCont} for that importance measure remains an issue for further study.\par
The present paper has shown that importance measures can be constructed from dependence measures. From a more fundamental perspective it seems to be worthwhile to further explore the connection between the notions of stochastic dependence and component importance, and establish links between them. This might prepare the ground for carrying over certain aspects from one framework to the other, such as ideas for axiomatization and classification of importance, and thereby facilitate a better understanding of this concept from a more abstract perspective. This, in turn, might lead to a better understanding of importance in reliability contexts.

\subsubsection*{Acknowledgments}
Thanks are due to the referees whose comments led to a significant improvement of the paper.

{\footnotesize

\end{document}